\documentclass[11pt,a4paper,english]{amsart}
\usepackage{babel}
\usepackage[latin1]{inputenc}
\usepackage{amsmath}
\usepackage{amsthm}
\usepackage{amsfonts}
\usepackage{indentfirst}
\usepackage{graphicx}

\usepackage{amssymb}
\usepackage[mathscr]{eucal}
\usepackage{tikz}

\usepackage{amsthm}
\usepackage{amscd}
\newtheorem{de}{Definition}
\newtheorem{pro}{Proposition}
\newtheorem{cor}{Corollary}
\newtheorem{teo}{Theorem}
\newtheorem{rem}{Remark}
\newtheorem{exa}{Example}

\newcommand{\F}{{\mathbb{F}}}

\title[Quantum codes from affine variety codes]{Quantum codes from affine variety codes and their subfield-subcodes}
\author{C. Galindo, F. Hernando}

\curraddr{\texttt{Carlos Galindo and Fernando Hernando:} Instituto
Universitario de Matem\'aticas y Aplicaciones de Castell\'on and
Departamento de Matem\'aticas, Universitat Jaume I, Campus de Riu
Sec. 12071 Castell\'{o} (Spain).} \email{{\rm Galindo:}
galindo@mat.uji.es; {\rm Hernando:} carrillf@mat.uji.es}
\date{}
\thanks{Supported by Spain Ministry of Economy
MTM2012-36917-C03-03 and University Jaume I: PB1-1B2012-04.}

\begin{document}

\begin{abstract}
We use affine variety codes and their subfield-subcodes for obtaining quantum stabilizer codes via the CSS code construction. With this procedure we get codes with good parameters, some of them exceeding the quantum Gilbert-Varshamov bound given by Feng and Ma.
\end{abstract}

\maketitle

\section*{Introduction}
Shor's algorithm \cite{22RBC} for factoring integers opens the possibility of breaking some cryptographical systems. This is a clear example of why scientists are interested in computers based on the principles of quantum mechanics. The fact that arbitrary quantum states cannot be replicated seemed to suggest that error correction could not be used on quantum mechanical systems \cite{26RBC}. However this is not true as showed in  \cite{23RBC}. Binary stabilizer codes are the most studied quantum error-correcting codes. There exist a lot of papers which consider them, for simplicity we only cite \cite{19kkk, 38kkk} as seminal works.

In this paper, we are interested in more general stabilizer codes defined over finite fields and constructed by using a class of linear error-correcting codes called affine variety ones. For us $q=p^r$ will be a positive integer power of a prime number $p$ and $\mathbb{C}^q$ the $q$-dimensional complex vector space representing the states of a quantum mechanical system. $|x\rangle$ will be the vectors of a distinguished orthonormal basis of $\mathbb{C}^q$, where $x \in \mathbb{F}_q$, $\mathbb{F}_q$ being the finite field with $q$ elements. By definition, a {\it quantum error-correcting code} will be a $s$-dimensional subspace of $ \mathbb{C}^{q^n} = \mathbb{C}^q\otimes \mathbb{C}^q \otimes \cdots \otimes \mathbb{C}^q$. Let $a, b \in \mathbb{F}_q$, then the unitary operators on $\mathbb{C}^q$, $X(a) |x\rangle = |x+a\rangle$ and $Z(b) |x\rangle = \beta^{\mathrm{tr}(bx)}|x\rangle$, where $\mathrm{tr}: \mathbb{F}_q \rightarrow \mathbb{F}_p $ is the trace map and $\beta$ a primitive $p$th root of unity, allow us  to consider the set $\varepsilon = \{ X(a) Z(b) | a, b \in \mathbb{F}_q\}$ of error operators. For $\mathbf{a}=(a_1, a_2, \ldots, a_n) \in \mathbb{F}_q^n$, define $X(\mathbf{a}) := X(a_1) \otimes X(a_2) \otimes \cdots X(a_n)$ and $Z(\mathbf{a})$ analogously and write $\varepsilon_n = \{  X(\mathbf{a}) Z(\mathbf{b}) | \mathbf{a}, \mathbf{b} \in \mathbb{F}_q^n \}$ a nice error basis on the complex space $ \mathbb{C}^{q^n}$. A {\it stabilizer code} $C$ is a non-zero  subspace of  $\mathbb{C}^{q^n}$ such that
$
C= \cap_{H \in \Delta} \{\mathbf{v} \in \mathbb{C}^{q^n} | H \mathbf{v} = \mathbf{v} \},
$
for some subgroup $\Delta$ of the group spanned by $\varepsilon_n$, $G_n$.

A stabilizer  code $C$ has minimum distance $d$ if, and only if, it can detect all errors in $G_n$ with weight less than $d$, where the weight  is the number of nonidentity tensor components, although some error of weight $d$ cannot be detected. We will say that a code $C$ as above is an $((n,s,d))_q$-code. When the code is an $((n,q^k,d))_q$-code, we will simply say that it is an $[[n,k,d]]_q$-code.  $C$  is said to be {\it pure to} $t$ whenever the group $\Delta$ does not contain non-scalar matrices whose weight is less than $t$ and $C$ is called {\it pure} whenever it is pure to its minimum distance.

As in the binary case, classical codes can be used to provide quantum codes. The following result gives a first link between them and it can be found in \cite[Corollary 19]{kkk} (see also \cite{19kkk, BE, AK}).
\begin{pro}
Assume the existence of an $[n,k,d]$ linear code $E$ over $\mathbb{F}_{q^2}$  such that the dual code of $E$, $E^{\perp_1}$, with respect to  the Hermitian inner product,  satisfies $E^{\perp_1} \subseteq E$. Then, there exists an $[[n,2k-n,\geq d]]_q$-quantum code over $\mathbb{F}_{q}$ which is pure to $d$.
\end{pro}
As customary in classical coding theory, we prefer to use the standard Euclidean inner product to provide quantum codes. The symbol $\perp$ will be used to represent dual spaces with respect to that inner product. So, in this paper we will use the so-called CSS code construction after the papers \cite{20kkk} and \cite{95kkk}. We summarize the idea in the next two results, which can be found as Lemma 20 and Corollary 21 in \cite{kkk}.

\begin{teo}
Let $C_1$ and $C_2$ two linear error-correcting block codes with parameters $[n,k_1,d_1]$ and $[n,k_2,d_2]$ over the field  $\mathbb{F}_{q}$ and such that $C_2^{\perp} \subseteq C_1$.  Then, there exists an $[[n, k_1+k_2 -n, d]]_q$ stabilizer code with minimum distance $$d = \min \left\{\mathrm{wt}(c) | c \in (C_1 \setminus C_2^\perp) \cup (C_2 \setminus C_1^\perp) \right\},$$ which is pure to $\min \{d_1,d_2\}$.
\end{teo}
For an explanation of how to get the codes, we refer to \cite[Theorem 13]{kkk} where the additive code $C_1^\perp \times C_2^\perp$ is used.
\begin{cor} \label{bueno}
Let $C$ be a linear $[n,k,d]$ error-correcting block code over  $\mathbb{F}_{q}$ such that $C^\perp \subseteq C$. Then, there exists an $[[n, 2k -n, \geq d]]_q$ stabilizer code which is pure to $d$.
\end{cor}

Quantum codes admit bounds on their parameters as the quantum singleton or the Hamming bound \cite{80kkk, ash, opt, feng, kkk}, which give necessary conditions for the existence of quantum codes. With the same philosophy of the classical Gilbert-Varshamov bound, a sufficient condition for the existence of pure stabilizer codes is given by Feng and Ma in \cite{feng}. Assume $n > 2$ and $d \geq 2$. When $k \geq 2$ and $n \equiv k \; (\mathrm{mod} \;2)$, that condition is
\[
 \sum_{i=1}^{d-1} (q^2 -1)^{i-1} \binom{n}{i}  < \frac{q^{n-k+2}-1}{q^2-1}.
\]
In case $n$ odd and $k=1$, the condition is $q^n +1 > \sum_{i=1}^{d-1} \binom{n}{i} [q (q^2-1)^{i-1} + (-1)^{i+1} (q + 1)^{i-1}]$ and there exists a similar formula for the case $n$ even and $k=0$.

The literature contains many quantum codes derived from classical ones. We only cite \cite{BE, kkk} and some  other recent papers  supported on BCH and quasi-cyclic codes \cite{guar1, cyclic,cyclic2,guar14, guar}. In  this paper, we consider affine variety codes to get stabilizer codes through the CSS code construction. We use Corollary \ref{bueno} for this purpose. The direct application of this procedure yields, in general,  codes whose supporting field can be large.
Subfield-subcodes \cite{delsarte, Stich, Fer, Fer1} are used to decrease it and we consider this type of codes for providing stabilizer  codes with good parameters. Our goal consists of finding  suitable affine variety codes over a field $\mathbb{F}_{p^r}$  that produce stabilizer codes over  $\mathbb{F}_{p^s}$ for some $s$ that divides $r$. This gives rise to better parameters since one can use codes over large fields to provide stabilizer codes over smaller finite fields. In this way, we get several examples of stabilizer codes improving some of those given in \cite{edel} and \cite{guar}. In addition, using our procedure, we are able to obtain some stabilizer codes that do not satisfy the above mentioned analogue to the classical Gilbert-Varshamov sufficient condition.


It is also worth to mention that affine variety codes  are, in some cases, a particular case of multivariable abelian codes. These codes have been studied in \cite{Edgar}, where for the case $q=4$, self-orthogonal and self-dual codes with respect to the trace inner product are characterized.

Section \ref{afin} of the paper introduces affine variety codes over a finite field $\mathbb{F}_{p^r}$, where $p$ is a prime number and $r$ a positive integer. The codes are given by evaluating the polynomials in several variables belonging to vector spaces generated by monomials whose exponents are in some subsets $U$ of the cartesian product
$$\{0, 1, \ldots,N_1-1\}\times \{0, 1, \ldots,N_2-1\} \times\cdots\times\{0, 1, \ldots,N_m-1\},$$ where the integer $N_i$ divides $p^r -1$ for all $i$, $1 \leq i \leq m$. We also devote Section \ref{afin} to show that good choices of sets $U$ give rise to stabilizer codes over $\mathbb{F}_{p^r}$. However, the supporting field of these codes is, in general, large. To reduce the size of the field, in Section \ref{dos} we introduce and study subfield-subcodes of our affine variety codes providing, in Theorem \ref{base}, a basis for the vector space of polynomials associated with affine variety codes but evaluating to a subfield $\mathbb{F}_{p^s}$ of $\mathbb{F}_{p^r}$.
Dual codes of the mentioned subfield-subcodes are treated in Section \ref{tres}. They are useful for proving our main result, Theorem \ref{main}, where we give conditions on the above sets $U$ for obtaining good stabilizer codes and parameters for them. Finally, Section \ref{cuatro} shows parameters of the above mentioned stabilizer codes constructed with our procedure. These codes improve some of the quantum codes included in \cite{edel} and \cite[Tables I and II]{guar}, and several of them exceed the Gilbert-Varshamov bound.

\section{Affine variety codes}  \label{afin}

We devote this section to introduce the class of classical error-correcting codes that we will use to yield stabilizer codes via the CSS code construction. Consider a finite field $\mathbb{F}_{p^r}$ where $r$ is a positive integer. Set $\mathbb{F}_{p^{r}}[X_1, X_2, \ldots, X_m]$ the ring of polynomials in $m$ variables over  the field $\mathbb{F}_{p^{r}}$ and pick $m$ positive integer numbers $N_1, N_2, \ldots, N_m$ such that $N_i\mid p^r-1$ for $i=1, 2, \ldots, m$. Consider the ideal of $\mathbb{F}_{p^{r}}[X_1, X_2, \ldots, X_m]$ generated by the set of polynomials $\{X_1^{N_1}-1, X_2^{N_2}-1,  \ldots, X_m^{N_m}-1\}$, which will be denoted by $I$. Let $Z(I)$ be the set of zeroes of $I$, its cardinality is $n:=\mathrm{card} (Z(I))$ and we set $Z(I) = \{P_1, P_2, \ldots,P_n\}$. Now, write $R:=\mathbb{F}_{p^{r}}[X_1, X_2, \ldots, X_m]/I$ and consider the evaluation map $\mathrm{ev}:  R \rightarrow \mathbb{F}_{p^{r}}^n$ which maps any function $f\in R$ to $\mathrm{ev}(f)=(f(P_1), f(P_2), \ldots,f(P_n))$, where $f$ also denotes any representative of its class.

The family of codes that we are going to define will be determined by certain linear subspaces of $R$. Since the polynomials $X_i^{N_i} -1$ have no multiple roots, our codes are of semi-simple type \cite{Edgar1}.

It is worthwhile to mention that $G=\{X_1^{N_1}-1, X_2^{N_2}-1,  \ldots, X^{N_m}-1\}$ is a Gr\"{o}bner basis of the ideal $I$ with respect to (say, some fixed) lexicographical ordering. Hence, we can choose a {\it canonical representative} of each class given by a polynomial $f$ which will be its reduction module $G$. Frequently, we will use the same notation for expressing a class in $R$ and its canonical representative. The following result will be useful.

\begin{pro}\label{te:isomorphims}
The above introduced evaluation map, {\rm ev}, is an isomorphism of $\mathbb{F}_{p^r}$-vector spaces.
\end{pro}
\begin{proof}
It follows from the fact that $R$ and $\mathbb{F}_{p^{r}}^n$ have the same cardinality and the evaluation map is surjective.
\end{proof}


Throughout this paper, $\mathcal{H}$ will denote the hypercube $$\mathcal{H} := \{0, 1, \ldots,N_1-1\}\times \{0, 1, \ldots,N_2-1\} \times\cdots\times\{0, 1, \ldots,N_m-1\}$$
and our codes will be defined by suitable subsets $U$ of $\mathcal{H}$. Each $U$ gives rise to the set $\{X_1^{u_1} X_2^{u_2}\cdots X_m^{u_m} | (u_1,u_2,\ldots,u_m)\in U\}$ of monomials in $\mathbb{F}_{p^{r}}[X_1, X_2, \ldots, X_m]$. This set provides elements in $R$ which generate a vector space over $\mathbb{F}_{p^{r}}$ that will be denoted $\mathbb{F}_{p^{r}}^U$. In particular $R=\mathbb{F}_{p^{r}}^\mathcal{H}$ as vector spaces. Now, we introduce the concept of affine variety code. For a more general definition see \cite{FL} or \cite{Geil-Affine}.

\begin{de}
Let $U$ be a set as above. The image of the restriction of the evaluation map $\mathrm{ev}$ to $\mathbb{F}_{p^{r}}^U$ is denoted by $C_{U}$ and constitutes the {\it affine variety code associated to $U$ over $\mathbb{F}_{p^{r}}$}.
\end{de}

As a consequence of the previous proposition, it holds that the restriction map $
\mathrm{ev}: \mathbb{F}_{p^{r}}^U \rightarrow \mathbb{F}_{p^{r}}^n
$ is injective and therefore $\dim(C_U)= \mathrm{card} (U)$.

Consider an element $\mathbf{u} =(u_1,u_2, \ldots, u_n) \in \mathcal{H}$ and set $\hat{\mathbf{u}} \in \mathcal{H}$ the element $\hat{\mathbf{u}} : =(\hat{u}_1,\hat{u}_2, \ldots, \hat{u}_n)$ defined by $\hat{u}_i= 0$ if $u_i=0$ and $\hat{u}_i=N_i-u_i$ otherwise. Next result generalizes a close one  for toric codes. It can be found in \cite{maria-michael} and \cite{diego}.

\begin{pro}\label{DualGTC}
Let $C_U$ be the affine variety code defined by $U\subset \mathcal{H}$, then $C_U^{\perp}$ is the affine variety code
defined by $U^{\perp}=  \mathcal{H} \setminus \{\hat{\mathbf{u}} | \mathbf{u} \in U\}$ .
\end{pro}

The following result will be useful for obtaining stabilizer (quantum) codes.

\begin{teo}
With the above notations, the inclusion $C_U\subset C_{U}^{\perp}$ happens if, and only if,  $\hat{\mathbf{u}}\notin U$ for all $\mathbf{u} \in U$.
\end{teo}
\begin{proof}
Proposition \ref{DualGTC} shows that the code $C_{U}^{\perp}$ is the evaluation by $\mathrm{ev}$ of
$\mathbb{F}_{p^{r}}^{U^{\perp}}$, where  $U^{\perp}=  \mathcal{H}\setminus \{\hat{\mathbf{u}} | \mathbf{u} \in U\}$, hence $U\subset U^{\perp}$ if and only if $\hat{\mathbf{u}}\notin U$ for all $\mathbf{u} \in U$, what concludes the proof.
\end{proof}

By using the CSS code construction, one can deduce the following result.

\begin{cor}
\label{afines}
Let $N_1, N_2, \ldots, N_m$ be positive integers such that $N_i$ divides $p^r -1$ for all index $i$, $p$ being a prime number and $r$ a positive integer. Let $U$ be a nonempty subset of the hypercube $\mathcal{H}$ satisfying that $\hat{\mathbf{u}}\notin U$ for all $\mathbf{u} \in U$. Then, from the affine variety code $C_U$, a stabilizer code can be obtained. The parameters for this code are $[[n,k,\geq d]]_{p^{r}}$, where $n=N_1 N_2 \cdots N_m$, $k=n-2 \; \mathrm{card}(U)$ and $d=d(C_U^{\perp})$.
\end{cor}

\begin{exa}
\label{luno}
{\rm With the above notation, set $R$ the ring $\mathbb{F}_{4}[x,y,z]/<x^3-1,y^3-1,z^3-1>$. This means that the corresponding codes will have length  $n=27$. Our set $U$ will be
\[
U=\left\{ (0, 0, 1),( 1, 1, 0 ),(0, 1, 1), (1, 1, 1), (1, 2, 0) (1, 0, 2 ), (0, 1, 2 ),(1, 1, 2), (1, 2, 1), (2, 1, 1) \right\}.
\]
Since $U$ satisfies the conditions in Corollary \ref{afines}, $U$ yields a $[[27, 7, 6]]_{4}$ stabilizer code. The distance has been computed with the computational algebra system Magma \cite{magma}.}
\end{exa}
\begin{exa}
{\rm Let us show other example. Consider the ring $R=\mathbb{F}_{8}[x,y]/<x^7-1,y^7-1>$ and the set
$
U= \left\{ (1, 3),( 4, 4 ),(1, 6 ),(5, 0), (1, 2 ) \right\}.
$
This gives a $[[49,39,4]]_{8}$ stabilizer code.}
\end{exa}

Codes  in the above examples are only samples. Next section is devoted to show that new codes with algebraic structure and good parameters can be obtained by considering subfield-subcodes.



\section{Subfield-Subcodes of Affine variety Codes} \label{dos}

As above, we write $R= \mathbb{F}_{p^{r}}[X_1, X_2, \dots, X_m]/I$, where $I$ is the ideal of the polynomial ring $\mathbb{F}_{p^{r}}[X_1, X_2, \dots, X_m]$ generated by the set of polynomials $\{X_1^{N_1}-1, X_2^{N_2}-1,  \ldots, X^{N_m}-1\}$. In addition, we consider a positive integer $s$ which divides $r$. We say that an element $f  \in R$ evaluates to $\mathbb{F}_{p^{s}}$ whenever $f(\boldsymbol{\alpha}) \in \mathbb{F}_{p^s}$ for all $\boldsymbol{\alpha}  \in Z(I)$. Now, define $\mathcal{T}: R\rightarrow R$ the map given by $\mathcal{T}(f) = f + f^{p^s} + \cdots + f^{p^{s(\frac{r}{s}-1)}}$. Also, set $\mathrm{tr}_r^s: \mathbb{F}_{p^r} \rightarrow \mathbb{F}_{p^s}$ given by $\mathrm{tr}_r^s(x)= x + x^{p^s} + \cdots + x^{p^{s(\frac{r}{s}-1)}}$ and extend it to $\mathbf{tr}: \mathbb{F}_{p^{r}}^n \rightarrow \mathbb{F}_{p^{s}}^n$ by applying $\mathrm{tr}_r^s$  coordinatewise. Afterwards, we will need the following results.

\begin{pro}\label{p:Tmap}
With the above notations and for any $f\in R$, it holds:
\begin{enumerate}
    \item $\mathcal{T}(af) = a \mathcal{T}(f)$, for all element $a\in \F_{p^s}$.
    \item $ \mathcal{T}(f)^{p^s} = \mathcal{T}(f^{p^s}) = \mathcal{T}(f)$.
    \item $\mathrm{ev}(\mathcal{T}(f)) = \mathbf{tr}(\mathrm{ev}(f))$.
    \item $\mathrm{ev}(\mathcal{T}(f)) = \mathbf{0}$ happens if, and only if, $\mathcal{T}(f) = 0$.
\end{enumerate}
\end{pro}

\begin{proof} Items (1), (2) and (3) follow  from the definition of $\mathcal{T}$, properties of finite fields and the fact that we are working modulo $I$. Item (4) holds because the map $\mathrm{ev}$ is an isomorphism.
\end{proof}

\begin{pro}\label{p:gprops} Let $g\in R$. Then, the following statements are equivalent.
\begin{enumerate}
    \item $g = \mathcal{T}(h)$ for some $h\in R$.
    \item $g^{p^s} = g$.
    \item $g$ evaluates to $\mathbb{F}_{p^s}$.
\end{enumerate}
\end{pro}

\begin{proof} For a start, suppose that for some $h\in R$, $g = \mathcal{T}(h)$. Then
$$g^{p^s} = \mathcal{T}(h)^{p^s} = \mathcal{T}(h) = g,$$
where the second equality follows from Proposition \ref{p:Tmap}. If $g^{p^s} = g$, then for any $\boldsymbol{\alpha}\in\F_{p^r}^m$, $g(\boldsymbol{\alpha})^{p^s} = g(\boldsymbol{\alpha})$ and so $g(\boldsymbol{\alpha})\in \F_{p^s}$. Lastly suppose that, for every $\boldsymbol{\alpha}  \in \F_{p^r}^m$, $g(\boldsymbol{\alpha})\in \F_{p^s}$. Since $\mathbf{tr}$ is surjective, we can consider $\mathbf{y} \in \F_{p^r}^n$ such that $\mathbf{tr} (\mathbf{y}) = \mathrm{ev}(g)$. Now, consider the class $h$ of a interpolating polynomial satisfying $\mathrm{ev}(h) = \mathbf{y}$, then, $$\mathrm{ev }(\mathcal{T}(h)) = \mathbf{tr}(\mathrm{ev}(h)) = \mathrm{ev} (g)$$ and the proof is concluded since $\mathrm{ev}$ is an isomorphism.
\end{proof}


Given a positive integer $t$, $\mathbb{Z}_t$ will stand for the quotient ring $\mathbb{Z}/t\mathbb{Z}$.

\begin{de} A subset $\mathfrak{I}$ of the cartesian product $\mathbb{Z}_{N_1}\times \mathbb{Z}_{N_2} \times \cdots\times\mathbb{Z}_{N_m}$ is called to be a {\it cyclotomic set} if  $\mathfrak{I}= p\cdot \mathfrak{I} : = \{p \cdot \boldsymbol{\alpha} \;| \; \boldsymbol{\alpha}\in \mathfrak{I}\}$. $\mathfrak{I}$ will be a {\it minimal cyclotomic set} if there is $\boldsymbol{\alpha} \in \mathbb{Z}_{N_1}\times \mathbb{Z}_{N_2} \times \cdots\times\mathbb{Z}_{N_m}$ such that every element of $\mathfrak{I}$ can be written $p^{s i } \cdot \boldsymbol{\alpha}$ for some integer $i$.
\end{de}

Fix a monomial ordering on $\mathbb{Z}_{\geq 0}^m$, where $\mathbb{Z}_{\geq 0}$ denotes the nonnegative integers. A minimal cyclotomic set $\mathfrak{I}$  will be represented by that element $\mathbf{a} \in \mathfrak{I}$ with smallest coordinates with respect to the  above fixed monomial ordering. Notice that we use the ordering for determining a unique representative of the set $\mathfrak{I}$ and in the proof of Theorem \ref{base}. Thus, we will set $\mathfrak{I}_\mathbf{a}:= \mathfrak{I}$ and $i_\mathbf{a} : = \mathrm{card}(\mathfrak{I}_\mathbf{a})$. Finally, the set of elements $\mathbf{a}$ representing minimal cyclotomic sets is denoted by $\mathcal{A}$.

For every $ \mathbf{a}=(a_1, a_2, \ldots, a_m)\in\mathcal{A}$, $i_\mathbf{a}$ is a divisor of $r$ and it holds that $a_i p^{s i_\mathbf{a}} \equiv a_i \, (\text{mod } N_i)$, $1 \leq i \leq m$. In addition, every cyclotomic set is a union of minimal cyclotomic sets and  the minimal cyclotomic sets constitute a partition of $\mathbb{Z}_{N_1}\times \mathbb{Z}_{N_2} \times \cdots\times\mathbb{Z}_{N_m}$. Recall that we denote by $f$ an element in $R$ and its canonical representative and we set $\mathrm{supp}(f)$ the support of the canonical representative. Then, any element $f\in R$ may be decomposed in a unique way as a sum of classes of polynomials with support included in minimal cyclotomic sets. That is to say, $f = \sum_{\mathbf{a} \in \mathcal{A}} f_\mathbf{a}$, where $\mathrm{supp} (f_\mathbf{a}) \subseteq \mathfrak{I}_\mathbf{a}$. We also notice that $\mathrm{supp} (\mathcal{T}(f_\mathbf{a})) \subseteq \mathfrak{I}_\mathbf{a}$.

Now define the function $\mathcal{T}_\mathbf{a} : R \rightarrow R$ as $\mathcal{T}_\mathbf{a}(f) = f + f^{p^s} + \cdots + f^{p^{s(i_\mathbf{a} -1)}}$ and set $X^\mathbf{a} = X_1^{a_1} X_2^{a_2} \cdots X_m^{a_m}$. Then we are ready to state and prove Theorem \ref{base}, which gives a basis for the vector space of elements in $R$ evaluating to the field $\mathbb{F}_{p^{s}}$. First, we provide two results that help us in our purpose.

\begin{pro}\label{LI}
Let  $f$ be an element in $R$ that evaluates to $\mathbb{F}_{p^s}$ with $\mathrm{supp}(f) \subseteq \mathfrak{I}_{\mathbf{a}}$ and consider a primitive element $\beta$ of $\mathbb{F}_{p^{s i_{\mathbf{a}}}}$. Then, $f$ can be expressed as  a linear combination with coefficients in $ \mathbb{F}_{p^s}$ of the elements in $R$ given by $\mathcal{S}_{\mathbf{a}}^\beta:=  \left\{  \mathcal{T}_{\mathbf{a}} (X^\mathbf{a}), \mathcal{T}_{\mathbf{a}} (\beta X^\mathbf{a}), \ldots, \mathcal{T}_{\mathbf{a}} (\beta^{i_\mathbf{a} -1 } X^\mathbf{a})    \right\}$.
\end{pro}
\begin{proof}
Since $\mathrm{supp}(f) \subseteq \mathfrak{I}_{\mathbf{a}}$ and $f^{p^s}=f$,  there is some $\alpha \in  \mathbb{F}_{p^r}$ such that  $f= \sum_{j=0}^{i_\mathbf{a}-1} (\alpha X^{\mathbf{a}})^{p^{js}}$.
Moreover $ \alpha^{{p}^{s i_{\mathbf{a}}}}  = \alpha$, which implies that $\alpha \in \mathbb{F}_{{p^{si_\mathbf{a}}}}$.

We know that $\{1,\beta, \ldots,\beta^{i_\mathbf{a}-1}\}$ is a basis of
$\mathbb{F}_{{p}^{si_\mathbf{a}}}$ over  $\mathbb{F}_{{p^s}}$, so
$\alpha=a_0+a_1\beta+\cdots+a_{i_\mathbf{a}-1}\beta^{i_\mathbf{a}-1}$,
with $a_i\in \mathbb{F}_{p^s}$ for all $i$.
Therefore,
\[ f= \sum_{j=0}^{i_\mathbf{a} - 1} \alpha^{p^{js}} X^{p^{js} \mathbf{a}}=  \sum_{j=0}^{i_\mathbf{a} - 1} X^{p^{js} \mathbf{a}} \left( \sum_{l=0}^{i_\mathbf{a} - 1} a_l \beta^l \right)^{p^{js}} \]
\[
= \sum_{l=0}^{i_\mathbf{a} - 1} a_l \left( \sum_{j=0}^{i_\mathbf{a} - 1} \beta^{lp^{js}} X^{p^{js} \mathbf{a}} \right) = \sum_{l=0}^{i_\mathbf{a} - 1} a_l \mathcal{T}_\mathbf{a} (\beta^l X^\mathbf{a}).
\]
\end{proof}

\begin{pro}\label{LinearlyInd}
 The vectors in the previous considered set $\mathcal{S}_{\mathbf{a}}^\beta$ are linearly independent over $\mathbb{F}_{p^s}$.
\end{pro}
\begin{proof}
Reasoning by contradiction, assume that $\sum_{l=0}^{i_\mathbf{a} - 1} a_l \mathcal{T}_\mathbf{a} (\beta^l X^\mathbf{a}) = 0$. Then, the term whose attached monomial is $X^\mathbf{a}$  and appears in the left hand side of the above equality is $(a_0+a_1\beta+\cdots+a_{i_\mathbf{a}-1}\beta^{i_\mathbf{a}-1})X^\mathbf{a}$  and it must vanish. This is true only if $\beta$ is a root of the univariate polynomial $a_0+a_1 Z+\cdots+a_{i_\mathbf{a}-1}Z^{i_\mathbf{a}-1}$. This gives the desired contradiction because the minimal polynomial of $\beta$ has degree $i_\mathbf{a}$.
\end{proof}
Next, we state the mentioned theorem.
\begin{teo}\label{base}
The following set
\[
\Omega_s^R := \bigcup_{\mathbf{a} \in \mathcal{A} } \left\{ \mathcal{T}_{\mathbf{a}} (\beta^{l} X^\mathbf{a}) \; | \;
0 \leq l \leq i_\mathbf{a}-1 \mbox{ and  $\beta$ is a primitive element of $\mathbb{F}_{p^{s i_\mathbf{a}}}$} \right\}
\]
constitutes a basis for the vector space  over $\mathbb{F}_{p^s}$ of elements in $R$ evaluating to $\mathbb{F}_{p^s}$.
\end{teo}
\begin{proof}
We start by proving that  the classes in $\Omega_s^R$ are linearly independent. This holds, on the one hand,  because there is no linear dependence among the elements in $\Omega_s^R$ supported on different minimal cyclotomic sets. Indeed, any monomial of any element supported on $\mathfrak{I}_{\mathbf{a}}$ is different from that of any other  supported on $\mathfrak{I}_{\mathbf{a'}}$ with $\mathbf{a} \neq \mathbf{a'}$. On the other hand Proposition  \ref{LinearlyInd} proves the independence of the elements supported on the same set $\mathfrak{I}_{\mathbf{a}}$, which shows our statement.

To conclude the proof, we are going to show that the set $\Omega_s^R$ generates the vector space  of elements $f$ in $R$ evaluating to $\mathbb{F}_{p^s}$. Recall that we are using canonical polynomials for representing their corresponding classes in $R$. Consider the term in $f$ with smallest order for the above mentioned monomial ordering on $\mathbb{Z}^m_{\geq 0}$, say $\beta^{k_1} X^{\mathbf{a}_1}$, then $\mathcal{T}_{\mathbf{a}_1} (\beta^{k_1} X^{\mathbf{a}_1}) = \sum_{l=0}^{i_{\mathbf{a}_1} - 1} (\beta^{k_1} X^{\mathbf{a}_1})^{p^{ls}}$ must appear in $f$ because it evaluates to $\mathbb{F}_{p^s}$. Since $\beta^{k_1} X^{\mathbf{a}_1}$ has the smallest order in $f$,  $\mathbf{a}_1$ must be one of the elements in $\mathcal{A}$. Assume, without loss of generality, that these elements are $\{\mathbf{a}_1, \mathbf{a}_2, \ldots, \mathbf{a}_{t}\}$. Set $f_1=f-\mathcal{T}_{\mathbf{a}_1} (\beta^{k_1} X^{\mathbf{a}_1})$ and pick its monomial with smallest order, say $\beta^{k_2} X^{\mathbf{a}_2}$. Again the polynomial $\mathcal{T}_{\mathbf{a}_2} (\beta^{k_2} X^{\mathbf{a}_2})$ must appear in $f_1$. We can repeat the above procedure and consider $f_2=f_1- \mathcal{T}_{\mathbf{a}_2} (\beta^{k_2} X^{\mathbf{a}_2})$. We will finish in $t$ steps and this will provide the desired expression of $f$ as a linear combination of elements in $\Omega_s^R$, which concludes the proof.
\end{proof}

We have just provided  a constructive way of obtaining all classes in $R$
that evaluate to $\mathbb{F}_{p^s}$. In particular, if we restrict to
those ones with support in $U$, we have a formula for
the dimension of an affine variety subfield-subcode.
\begin{teo}
\label{t:D_U}
Let $U$ be a subset of $\{0, 1, \dots,N_1-1\}\times \{0, 1, \dots,N_2-1\} \times \cdots\times\{0, 1, \dots,N_m-1\}$, $n=N_1 N_2 \cdots N_m$  and define $C_U^s : = C_U \cap \mathbb{F}_{p^s}^n$. Then,
\[C_U^s = \mathrm{ev} \; \left( \mathcal{T}(\mathbb{F}_{p^r}^{\mathcal{H}})\cap\mathbb{F}_{p^r}^U \right),
\]
$C_U^s$ is generated by the images under the evaluation map {\rm ev} of the following elements in $R$
\[
\bigcup_{\mathfrak{I}_\mathbf{a}| \mathfrak{I}_\mathbf{a}\subseteq U}  \left\{ \mathcal{T}_{\mathbf{a}} (\beta^{l} X^\mathbf{a}) \; | \; 0 \leq l \leq i_\mathbf{a}-1 \mbox{ and  $\beta$ is a primitive element of $\mathbb{F}_{p^{s i_\mathbf{a}}}$} \right\}
\]
and
\[ \dim C_U^s = \sum_{\mathfrak{I}_\mathbf{a}| \mathfrak{I}_\mathbf{a}\subseteq U} i_\mathbf{a}.
\]
\end{teo}

\section{Quantum codes from  subfield-subcodes of affine variety codes} \label{tres}

We devote this section to explain which of our affine variety subfield-subcodes yield stabilizer codes. Notice that our quantum codes will be defined over a {\it small} field $\mathbb{F}_{p^s}$ with the advantage that the original code is supported over a {\it large} field $\mathbb{F}_{p^r}$.

Previously to state our main result, we will need to describe dual codes of subfield-subcodes. From Proposition \ref{DualGTC}, we know that $C_U^{\perp}$ is the affine variety code defined by $U^{\perp}$. Then, we get
\[
\left(C_U^s\right)^{\perp}= \mathbf{tr} \left(C_{U^{\perp}}\right) = \mathbf{tr} \left(\mathrm{ev} (\mathbb{F}_{p^r}^{U^\perp})\right) = \mathrm{ev} \left(\mathcal{T}(\mathbb{F}_{p^r}^{U^{\perp}})\right),
\]
where the first equality follows from Delsarte's Theorem \cite{delsarte} and the last one from the fact that, by Proposition \ref{p:gprops}, the following map composition equality $\mathrm{ev} \circ \mathcal{T} = \mathbf{tr} \circ \mathrm{ev}$ holds. Notice that, as above, we are identifying classes in $R$ with canonical representatives. Now, $\mathcal{T}( \mathbb{F}_{p^r}^{U^{\perp}} )$ is generated by $\mathcal{T}(\gamma X^\mathbf{a})$ for $\mathbf{a}\in U^{\perp}$ and $\gamma \in \mathbb{F}_{p^r}$.  If one fixes  $\mathbf{a}$ and varies $\gamma$ over the field, then the set $\{\mathcal{T}_{\mathbf{a}} (\beta^l X^\mathbf{a})\}_{0 \leq l \leq i_\mathbf{a}-1}$, $\beta$ primitive, is obtained. Thus we have proved the following result.

\begin{teo}\label{DualOfSFSC}
Let $U \subseteq \mathcal{H}$ be as in Section \ref{afin}. The  dual code $(C_U^s)^\perp$ of the code $C_U^s$ is generated by the image by {\rm ev} of the following elements in $R$:
\[
\bigcup_{\mathfrak{I}_\mathbf{a}| \mathfrak{I}_\mathbf{a} \cap U^\perp \neq \emptyset}  \left\{ \mathcal{T}_{\mathbf{a}} (\beta^{l} X^\mathbf{a}) \; | \; 0 \leq l \leq i_\mathbf{a}-1 \mbox{ and  $\beta$ is a primitive element of $\mathbb{F}_{p^{s i_\mathbf{a}}}$} \right\}.
\]
As a consequence, it holds that
\[ \dim (C_U^s)^\perp = \sum_{\mathfrak{I}_\mathbf{a}| \mathfrak{I}_\mathbf{a} \cap U^\perp \neq \emptyset} i_\mathbf{a}.
\]
\end{teo}

We conclude this section with our main result which allows us to construct good stabilizer codes.
We will need the concept of complementary of a minimal  cyclotomic set $\mathfrak{I}_\mathbf{a}$. That set is the subset of $\mathbb{Z}_{N_1}\times \mathbb{Z}_{N_2} \times \cdots\times\mathbb{Z}_{N_m}$ defined as
$\hat{\mathfrak{I}_\mathbf{a}}:=\left \{\hat{\mathbf{u}}\; | \; \mathbf{u} \in \mathfrak{I}_\mathbf{a} \right\}$.

\begin{teo} \label{main}
Let $p$ be a prime number and $r$ and $s$  positive integers such that $s|r$. Let $$\mathcal{H} := \{0, 1, \ldots,N_1-1\}\times \{0, 1, \ldots,N_2-1\} \times\cdots\times\{0, 1, \ldots,N_m-1\}$$ be the hypercube defined in Section \ref{afin}, where $N_i$ divides $p^r -1$ for all index $i$, and consider a nonempty subset $U$ of $\mathcal{H}$. Then,
\begin{enumerate}
\item The codes' inclusion $C_U^s\subseteq (C_{U}^s)^{\perp}$ happens if, and only if, $\hat{\mathfrak{I}_\mathbf{a}}$ is not contained in $U$ whenever $\mathfrak{I}_\mathbf{a}$ is.
\item Assume that $U$ satisfies the conditions in the previous item. Then, from the affine variety code $C_U^s$ a stabilizer code can be obtained. The parameters for that code are $[[n,k, \geq d]]_{p^{s}}$, where $n=N_1 N_2 \cdots N_m$, $k=n-2 \; \sum_{\mathfrak{I}_\mathbf{a} | \mathfrak{I}_\mathbf{a} \subseteq U} i_{\mathbf{a}}$ and $d=d\left((C_U^s)^{\perp}\right)$.
\end{enumerate}
\end{teo}
\begin{proof}
Theorem \ref{DualOfSFSC} proves that the dual code $(C_U^s)^{\perp}$ is given by the evaluation of the elements $\mathcal{T}_{\mathbf{a}} (\beta^{l} X^\mathbf{a})$, where $0 \leq l \leq i_\mathbf{a}-1$ and $\beta$ is a primitive element of $\mathbb{F}_{p^{s i_\mathbf{a}}}$, whenever $\mathfrak{I}_\mathbf{a} \cap U^\perp \neq \emptyset$. Therefore, $C_U^s\subseteq (C_U^s)^{\perp}$
if and only if  $\mathfrak{I}_\mathbf{a} \cap U^\perp\ne \emptyset$   when $\mathfrak{I}_\mathbf{a} \subseteq U$ and this happens if, and only if, the complementary set of $\mathfrak{I}_\mathbf{a}$ satisfies $\hat{\mathfrak{I}_\mathbf{a}}\not\subseteq U$ for all minimal cyclotomic set $\mathfrak{I}_\mathbf{a} \subseteq U $. This proves our first assertion. Second one follows from CSS code construction.
\end{proof}

\begin{rem}
{\rm
One of the best  known classes of affine variety codes are those where the ring $R$ (it could be more general than that in this paper) admits a weight function. This function takes values in an ordered semigroup and gives a suitable nested sequence of vector spaces in $R$, $L_1 \subset L_2 \subset \cdots \subset L_r$, whose dimensions increase in one unit. The evaluation of these spaces provides a nested family of primary codes $C_1 \subset C_2 \subset \cdots \subset C_r$ and their corresponding dual ones $(C_r)^{\perp} \subset (C_{r-1})^{\perp} \subset \cdots \subset (C_1)^{\perp}$. The weight function allows us to define the so-called Feng-Rao bound on the minimum distance of the previous codes. In this paper, from the point of view of classical codes and despite not having a weight function, we show a way of getting suitable sets $U_i \subset \mathcal{H}$, $1 \leq i \leq r$, to obtain nested sequences of codes $C_{U_1} \subset C_{U_2} \subset \cdots \subset C_{U_r}$ (respectively, $C_{U_1}^s \subset C_{U_2}^s \subset \cdots \subset C_{U_r}^s$) such that  $C_{U_r} \subset C_{U_r}^\perp $ (respectively, $C_{U_r}^s \subset (C_{U_r}^s)^\perp $) and therefore $(C_{U_r})^{\perp} \subset (C_{U_{r-1}})^{\perp} \subset \cdots \subset (C_{U_1})^{\perp}$ (respectively, $(C_{U_r}^s)^{\perp} \subset (C_{U_{r-1}}^s)^{\perp} \subset \cdots \subset (C_{U_1}^s)^{\perp}$ ). Our main result is to determine the  dimensions of the above mentioned codes.

It would be interesting to know an explicit formula or tight bound for the value $d((C_U^s)^{\perp})$, $ U \subset \mathcal{H}$. Reasoning as in Theorem \ref{main}, when the inclusion $ (C_{U}^s)^{\perp} \subseteq C_U^s $ holds, one can get an $[[n,k,d]]_{p^s}$ code where $n$ is as above, $k=2 \; \sum_{\mathfrak{I}_\mathbf{a} | \mathfrak{I}_\mathbf{a} \subseteq U} i_{\mathbf{a}} - n $ and $d=d(C_U^s)$. In this case a lower bound for the distance $d$ can be described. Indeed, $ d(C_U^s)\geq d(C_U)$ and a lower bound for the distance of the code $C_U$ can be computed following the procedure given in \cite{Geil-Martin}.

Codes in this paper are constructed by applying the CSS code construction to suitable linear codes. However, to get quantum codes, one can also use another bilinear pairings as the trace alternating inner product. As a referee of this paper pointed out, to relate our codes with those obtained with respect to different bilinear pairings is an interesting future work.
}
\end{rem}

\section{Examples} \label{cuatro}
In our examples, sets $U$ will be the union of different minimal cyclotomic sets, since otherwise the evaluation will not be in $\mathbb{F}_{p^s}^n$. Next, we will show a table containing parameters corresponding with some quantum codes coming from  subfield-subcodes of affine variety codes. Notice that the value $d$ in the table is a lower bound of the actual minimum distance. These codes either improve or add new parameters with respect to those ones given in \cite{edel}. In addition, some of them have a symbol GV which means that exceed the quantum Gilbert-Varshamov bound  \cite[Theorem 1.4 and Corollary 2.3]{feng}. Finally, codes with a symbol L improve the parameters of some codes in \cite[Tables I and II]{guar}.
Computations has been done by writing a Magma \cite{magma} function. After the table, the reader can find, also in tabular form, those {\it subsets} $U$ and {\it values} $N_i$, $p$, $r$ and $s$ providing the codes. For simplicity, a code given by $U$ is also called $U$.
\vspace{2mm}

\begin{center}
\begin{tabular}{||c|c|c|c|c|c||c|c|c|c|c|c||}
  \hline \hline
  Code  & Symbol &$n$ & $k$ & $d$ & $q=p^s$ & Code  & Symbol & $n$ & $k$ & $d$ & $q=p^s$ \\
  \hline \hline
  $U_1$ & GV & 200 & 184 & 4 & 3 & $U_2$ & & 147 & 127 & 3 & 2 \\
  \hline
  $U_3$ & & 147 & 123 & 4 & 2 &  $U_4$ & & 147 & 105 & 6 & 2 \\
  \hline
  $U_5$ & GV & 23 & 1 & 7 & 2 & $U_6$& & 189 & 147 & 5 & 2 \\
  \hline
  $U_7$ & & 189 & 141 & 6 & 2 & $U_8$ & & 189 & 129 & 7 & 2 \\
  \hline
  $U_9$ & & 217 & 171 & 6 & 2 & $U_{10}$ & & 245 & 209 & 4 & 2 \\
  \hline
  $U_{11}$ & & 245 & 179 & 6 & 2 & $U_{12}$ & & 441 & 411 & 4 & 2 \\
  \hline
  $U_{13}$ & & 441 & 423 & 3 & 2 & $U_{14}$ && 21 & 13 & 3 & 4 \\
  \hline
  $U_{15}$ & & 21 & 7 & 5 & 4 & $U_{16}$ & & 35 & 5 & 7 & 4 \\
  \hline
  $U_{17}$ & & 45 & 25 & 6 & 4 & $U_{18}$ & & 75 & 67 & 3 & 4 \\
  \hline
  $U_{19}$ & & 27 & 19 & 3 & 4 & $U_{20}$ & & 15  & 9 & 3  & 4 \\
  \hline
  $U_{21}$ & &  45 & 33 & 4  &4 & $U_{22}$ & & 147  & 137 & 3  & 4 \\
  \hline
  $U_{23}$ & &  73 & 67 & 3 & 8 & $U_{24}$ &L& 73  & 55 & 6  & 8\\
  \hline
  $U_{25}$ & &  21 & 9 & 5 & 8 & $U_{26}$ &GV& 63  & 55 & 4  & 8\\
  \hline
  $U_{27}$ & &  63 & 49 & 5 & 8 & $U_{28}$ & L& 73  & 43 & 8  & 8\\
  \hline
  $U_{29}$ & &  64 & 48 & 4 & 3 & $U_{30}$ & & 64  & 52 & 3  & 3\\
  \hline
  $U_{31}$ &GV, L& 11  & 1 & 5  & 3 &  $U_{32}$ & L &  71  & 51 & 5  & 5\\
  \hline
  $U_{33}$ &L &  31 & 13 & 6 & 5 & $U_{34}$ & &  71  & 41 & 8  & 5\\
  \hline
  $U_{35}$ & &  96 & 88 & 3 & 5 & $U_{36}$ & & 96  & 84 & 4  & 5\\
  \hline
  $U_{37}$ &  &  36 & 30 & 3 & 7 & $U_{38}$ & & 36  & 26 & 4  & 7\\
  \hline
  \hline
\end{tabular}
\end{center}
\vspace{2mm}

\begin{center}
\begin{tabular}{||c|c|c|c|c|c|c||}
  \hline \hline
  Code / Subset & $p$ & $r$ & $s$ & $N_1$  & $N_2$ & $N_3$ \\
  \hline \hline
  $
  \begin{array}{r}
  U_1= \{ ( 3, 3, 3 ),( 4, 4, 1 ),( 2, 2, 3 ),( 1, 1, 1 ), ( 3, 0, 6 ),( 4, 0, 2 ),\\
   ( 2, 0, 6 ), ( 1, 0, 2 )\}\\
   \end{array}
   $
    & 3 & 4 & 1 & 5 & 5 &8\\
  \hline
  $
  \begin{array}{r}
   U_2= \{ ( 2, 4, 2 ),( 4, 1, 1 ),( 1, 2, 2 ),( 2, 4, 1 ), ( 4, 1, 2), ( 1, 2, 1 ),\\ ( 2, 0, 0 ),( 4, 0, 0 ), ( 1, 0, 0 ),( 2, 5, 0 ),( 4, 3, 0 ),( 1, 6, 0 ) \}\\
   \end{array}
   $
    & 2 & 6 & 1 & 7 & 7 &3\\
   \hline
  $
  \begin{array}{r}
   U_3= \{ ( 6, 2, 2 ),( 5, 4, 1 ),( 3, 1, 2 ),( 6, 2, 1 ),( 5, 4, 2 ),( 3, 1, 1 ),\\ ( 2, 3, 0 ),( 4, 6, 0 ),( 1, 5, 0 ),( 6, 0, 0 ),( 5, 0, 0 ),( 3, 0, 0 ) \}\\
   \end{array}
   $
    & 2 & 6 & 1 & 7 & 7 &3\\
  \hline
  $
  \begin{array}{r}
   U_4= \{ ( 2, 4, 0 ),( 4, 1, 0 ),( 1, 2, 0 ),( 6, 0, 2 ),( 5, 0, 1 ),( 3, 0, 2 ),\\ ( 6, 0, 1 ),( 5, 0, 2 ),( 3, 0, 1 ),( 6, 2, 0 ),( 5, 4, 0 ),( 3, 1, 0 ),\\ ( 0, 6, 2 ),( 0, 5, 1 ),( 0, 3, 2 ),( 0, 6, 1 ), ( 0, 5, 2 ),( 0, 3, 1 ),\\ ( 2, 0, 0),( 4, 0, 0 ),( 1, 0, 0 )\}\\
   \end{array}
   $
    & 2 & 6 & 1 & 7 & 7 &3\\
    \hline
    $
  \begin{array}{r}
   U_5= \{ 2, 4, 8, 16, 9, 18, 13, 3, 6, 12, 1 \}\\
   \end{array}
   $
    & 2 & 11 & 1 & 23 & - &-\\
   \hline
  \hline
   \end{tabular}
\end{center}
\vspace{2mm}

    \begin{center}
\begin{tabular}{||c|c|c|c|c|c|c||}
  \hline \hline
  Code / Subset & $p$ & $r$ & $s$ & $N_1$  & $N_2$ & $N_3$ \\
  \hline \hline
    $
  \begin{array}{r}
   U_6= \{ ( 0, 2, 0 ),( 0, 4, 0 ),( 0, 1, 0 ),( 0, 6, 2 ),( 0, 5, 1 ),( 0, 3, 2 ),\\ ( 0, 6, 1 ),( 0, 5, 2 ),( 0, 3, 1 ),( 2, 6, 2 ),( 4, 5, 1 ),( 8, 3, 2 ),\\  (7, 6, 1 ),( 5, 5, 2 ),( 1, 3, 1 ),( 2, 1, 0 ),( 4, 2, 0 ),( 8, 4, 0 ),\\( 7, 1, 0 ),(5, 2, 0 ),( 1, 4, 0 )\}\\
   \end{array}
   $
    & 2 & 6 & 1 & 9 & 7 &3\\
    \hline
    $
  \begin{array}{r}
   U_7= \{ ( 2, 4, 1 ),( 4, 1, 2 ),( 8, 2, 1 ),( 7, 4, 2 ),( 5, 1, 1 ),( 1, 2, 2 ),\\ ( 0, 6, 2 ),(0,5,1),  ( 0, 3, 2 ),( 0, 6, 1 ),( 0, 5, 2 ),( 0, 3, 1 ),\\ (2, 6, 2 ),( 4, 5, 1 ),( 8, 3, 2 ),( 7, 6, 1 ),( 5, 5, 2 ),( 1, 3, 1 ),\\ ( 2, 2, 0 ),( 4, 4, 0 ),( 8, 1, 0 ), ( 7, 2, 0 ), ( 5, 4, 0 ), ( 1, 1, 0 )\}\\
   \end{array}
   $
    & 2 & 6 & 1 & 9 & 7 &3\\
    \hline
    $
  \begin{array}{r}
   U_8= \{ ( 2, 3, 0 ),( 4, 6, 0 ),( 8, 5, 0 ),( 7, 3, 0 ),( 5, 6, 0 ),( 1, 5, 0 ),\\ ( 6, 6, 0 ),( 3, 5, 0 ),( 6, 3, 0 ),( 3, 6, 0 ),( 6, 5, 0 ),( 3, 3, 0),\\ ( 2, 3, 1 ),( 4, 6, 2 ),( 8, 5, 1 ),( 7, 3, 2 ),( 5, 6, 1 ),( 1, 5, 2 ),\\ (2, 4, 2 ),( 4, 1, 1 ),( 8, 2, 2 ),( 7, 4, 1 ),( 5, 1, 2 ),( 1, 2, 1 ),\\  (6, 2, 2 ),( 3, 4, 1 ),( 6, 1, 2 ),( 3, 2, 1 ),( 6, 4, 2 ),( 3, 1, 1 )\}\\
   \end{array}
   $
    & 2 & 6 & 1 & 9 & 7 &3\\
    \hline
    $
  \begin{array}{r}
   U_9= \{( 0, 2 ),( 0, 4 ),( 0, 1 ),( 14, 0 ),( 28, 0 ),( 25, 0 ),( 19, 0 ), \\ ( 7, 0 ),( 22, 6 ),( 13, 5 ),( 26, 3 ),( 21, 6 ),( 11, 5 ), (22, 3 ), ( 13, 6 ),\\ ( 26, 5 ),( 21, 3 ),( 11, 6 ),( 22, 5 ),( 13, 3 ),( 26, 6 ), (21, 5 ),( 11, 3 )\}\\
   \end{array}
   $
    & 2 & 15 & 1 & 31 & 7 &-\\
    \hline
    $
  \begin{array}{r}
   U_{10}= \{ ( 2, 4, 0 ),( 4, 1, 0 ),( 1, 2, 0 ),( 6, 0, 2 ),( 5, 0, 4 ),( 3, 0, 3 ),\\ ( 6, 0, 1 ),( 5, 0, 2 ),( 3, 0, 4 ),( 6, 0, 3 ),( 5, 0, 1 ),( 3, 0, 2 ),\\  (6, 0, 4 ),( 5, 0, 3 ),( 3, 0, 1 ),( 2, 0, 0 ),( 4, 0, 0 ), (1, 0, 0 )\}\\
   \end{array}
   $
    & 2 & 12 & 1 & 7 & 7 &5\\
    \hline
  \hline
   \end{tabular}
\end{center}
\vspace{2mm}

    \begin{center}
\begin{tabular}{||c|c|c|c|c|c|c||}
\hline \hline
  Code / Subset & $p$ & $r$ & $s$ & $N_1$  & $N_2$ & $N_3$ \\
  \hline \hline
    $
  \begin{array}{r}
   U_{11}= \{ ( 2, 4, 2 ),( 4, 1, 4 ),( 1, 2, 3 ),( 2, 4, 1 ),( 4, 1, 2 ),( 1, 2, 4 ),\\ ( 2, 4, 3 ),( 4, 1, 1 ),( 1, 2, 2 ),( 2, 4, 4 ),( 4, 1, 3 ),( 1, 2, 1 ),\\ ( 6, 2, 0 ),( 5, 4, 0 ),( 3, 1, 0 ),( 6, 4, 2 ),( 5, 1, 4 ),( 3, 2, 3 ),\\( 6, 4, 1 ),( 5, 1, 2 ),( 3, 2, 4 ),( 6, 4, 3 ),( 5, 1, 1 ),( 3, 2, 2 ),\\( 6, 4, 4 ),( 5, 1, 3 ),( 3, 2, 1 ),( 6, 6, 0 ), ( 5, 5, 0 ),( 3, 3, 0 ),\\( 6, 0, 0 ),( 5, 0, 0 ),( 3, 0, 0 ) \}\\
   \end{array}
   $
    & 2 & 12 & 1 & 7 & 7 &5\\\hline
    $
  \begin{array}{r}
   U_{12}= \{ ( 0, 6, 2 ),( 0, 5, 4 ),( 0, 3, 1 ),
( 6, 2, 2 ),( 3, 4, 4 ),( 6, 1, 1 ),\\( 3, 2, 2 ),( 6, 4, 4 ),( 3, 1, 1 ),( 2, 2, 3 ),( 4, 4, 6 ),( 8, 1, 5 ),\\( 7, 2, 3 ),( 5, 4, 6 ), ( 1, 1, 5 ) \}\\
   \end{array}
   $
    & 2 & 12 & 1 & 9 & 7 &7\\\hline
    $
  \begin{array}{r}
   U_{13}= \{ ( 2, 0, 1 ),( 4, 0, 2 ),
( 8, 0, 4 ), ( 7, 0, 1 ),( 5, 0, 2 ), ( 1, 0, 4 ),\\( 0, 2, 0 ),( 0, 4, 0 ), ( 0, 1, 0 )\}\\
   \end{array}
   $
    & 2 & 12 & 1 & 9 & 7 &7\\\hline
    $
  \begin{array}{r}
   U_{14}= \{  0, 2 ),( 5, 1 ),( 6, 1 ),( 3, 1 )\}\\
   \end{array}
   $
    & 2 & 6 & 2 & 7 & 3 &-\\\hline
    $
  \begin{array}{r}
   U_{15}= \{ ( 0, 2 ),( 4, 1 ),( 2, 1 ),( 1, 1 ),( 5, 0 ),( 6, 0 ),( 3, 0 ) \}\\
   \end{array}
   $
    & 2 & 6 & 2 & 7 & 3 &-\\\hline
    $
  \begin{array}{r}
   U_{16}= \{ ( 4, 4 ),( 2, 1 ),( 1, 4 ),( 4, 1 ),( 2, 4 ),( 1, 1 ),( 5, 0 ),( 6, 0 ),\\( 3, 0 ),( 5, 3 ),( 6, 2 ),( 3, 3 ),( 5, 2 ),( 6, 3 ),( 3, 2 ) \}\\
   \end{array}
   $
    & 2 & 12 & 2 & 7 & 5 &-\\
    \hline
    $
  \begin{array}{r}
   U_{17}= \{ ( 3, 1, 1 ),( 2, 1, 1 ),( 0, 0, 2 ),( 4, 1, 2 ),( 1, 1, 2 ),( 4, 2, 0 ),\\ ( 0, 2, 2 ),( 3, 0, 1 ),( 2, 0, 1 ) \}\\
   \end{array}
   $
    & 2 & 4 & 2 & 5 & 3 &3\\\hline
    $
  \begin{array}{r}
   U_{18}= \{ ( 3, 0, 1 ),( 2, 0, 1 ),( 3, 1, 2 ),( 2, 4, 2 ) \}\\
   \end{array}
   $
    & 2 & 4 & 2 & 5 & 5 &3\\\hline
    $
  \begin{array}{r}
   U_{19}= \{ (2, 2, 0 ),( 0, 1, 2),( 2, 0, 1 ),( 1, 0, 0 ) \}\\
   \end{array}
   $
    & 2 & 2 & 2 & 3 & 3 &3\\ \hline
     $
  \begin{array}{r}
   U_{20}= \{ ( 0, 2 ),( 3, 1 ),( 2, 1 ) \}\\
   \end{array}
   $
    & 2 & 4 & 2 & 5 & 3 &-\\
  \hline
  \hline
\end{tabular}
\end{center}
\vspace{2mm}

 \begin{center}
\begin{tabular}{||c|c|c|c|c|c|c||}
\hline \hline
  Code / Subset & $p$ & $r$ & $s$ & $N_1$  & $N_2$ & $N_3$ \\
  \hline \hline
    $
  \begin{array}{r}
   U_{21}= \{ ( 4, 0, 2 ),( 1, 0, 2 ),( 0, 0, 1 ),( 0, 1, 0 ),( 0, 1, 1 ),( 0, 2, 1 )\}\\
   \end{array}
   $
    & 2 & 4 & 2 & 5 & 3 &3\\\hline
    $
  \begin{array}{r}
   U_{22}= \{ (4, 0, 1 ),( 2, 0, 1 ),( 1, 0, 1 ),( 4, 5, 2 ),( 2, 6, 2 ),( 1, 3, 2 ) \}\\
   \end{array}
   $
    & 2 & 6 & 2 & 7 & 7 &3\\\hline
    $
  \begin{array}{r}
   U_{23}= \{ 16, 55, 2 \}\\
   \end{array}
   $
    & 2 & 9 & 3 & 73 & - &-\\ \hline
    $
  \begin{array}{r}
   U_{24}= \{ 16, 55, 2,32, 37,4 , 53, 59, 34 \}\\
   \end{array}
   $
    & 2 & 9 & 3 & 73 & - &-\\
    \hline
    $
    \begin{array}{r}
   U_{25}= \{ (5, 2),(5, 1), (1, 2),
        (1, 1), (2, 0), (4, 0)\}\\
   \end{array}
   $
    & 2 & 6 & 3 & 7 & 3 &-\\ \hline
    $
  \begin{array}{r}
   U_{26}= \{ (3, 0), (6, 0), (1, 5), (1, 4)\}\\
   \end{array}
   $
    & 2 & 6 & 3 & 7 & 9 &-\\
    \hline
    $
   \begin{array}{r}
   U_{27}= \{ (6, 6 ), (6, 3), (3, 0), (1, 8 ),  ( 1, 1 ), ( 5, 5 ), ( 5, 4 )\}\\
   \end{array}
   $
    & 2 & 6 & 3 & 7 & 9 &-\\
    \hline
    $
    \begin{array}{r}
   U_{28}= \{ (23, 38, 12, 22, 30, 21,54, 67, 25, 56, 10, 7, 15, 47, 11 \}\\
   \end{array}
   $
    & 2 & 9 & 3 & 73 & - &-\\
    \hline
    $
   \begin{array}{r}
   U_{29}= \{ ( 3, 2 ),( 1, 6 ),( 0, 3 ),( 0, 1 ),( 7, 0 ),( 5, 0 ),( 7, 3 ),( 5, 1 ) \}\\
   \end{array}
   $
    & 3 & 2 & 1 & 8 & 8 &-\\ \hline
    $
  \begin{array}{r}
   U_{30}= \{ ( 3, 0 ),( 1, 0 ),( 6, 5 ),( 2, 7 ),(7, 7),( 5, 5 )\}\\
   \end{array}
   $
    & 3 & 2 & 1 & 8 & 8 &-\\
    \hline
    $
    \begin{array}{r}
   U_{31}= \{ (3,9,5,4,1\}\\
   \end{array}
   $
    & 3 & 5 & 1 & 11 & - &-\\
    \hline
  $
    \begin{array}{r}
   U_{32}= \{ 39, 53, 52, 47, 22 , 65, 41, 63, 31, 13 \}\\
   \end{array}
   $
    & 5 & 5 & 1 & 71 & - &-\\
    \hline
    $
  \begin{array}{r}
   U_{33}= \{ 15, 13, 3, 20, 7, 4, 29, 21, 12 \}\\
   \end{array}
   $
    & 5 & 3 & 1 & 31 & - &-\\
    \hline
    $
  \begin{array}{r}
   U_{34}= \{ 64, 36, 38, 48, 27, 15, 4, 20, 29, 3, 45, 12, 60, 16, 9\}\\
   \end{array}
   $
    & 5 & 5 & 1 & 71 & - &-\\
    \hline
    $
  \begin{array}{r}
   U_{35}= \{ (18, 1), (17, 0), (13, 0 ), (6, 0) \}\\
   \end{array}
   $
    & 5 & 2 & 1 & 24 & 4 &-\\
    \hline
    $
  \begin{array}{r}
   U_{36}= \{ (23, 3), (19, 3), (11, 2), (7, 2), (18, 0), (12, 3) \}\\
   \end{array}
   $
    & 5 & 2 & 1 & 24 & 4 &-\\
    \hline
     $
  \begin{array}{r}
   U_{37}= \{ (2, 1), (0, 4), (1, 3) \}\\
   \end{array}
   $
    & 7 & 2 & 1 & 6 & 6 &-\\
    \hline
     $
  \begin{array}{r}
   U_{38}= \{ (2, 0), (2, 2), (0, 5), (1, 1), (1, 2)\}\\
   \end{array}
   $
    & 7 & 2 & 1 & 6 & 6 &-\\
  \hline
  \hline
\end{tabular}
\end{center}

\end{document}